\documentclass[a4paper,11pt]{article}
\usepackage{amsthm}
\usepackage{amsmath}
\usepackage{amssymb}
\usepackage{fullpage}
\usepackage{rotating}

\newtheorem{definition}{Definition}
\newtheorem{theorem}{Theorem}
\newtheorem{lemma}{Lemma}
\newtheorem{proposition}[theorem]{Proposition}

\theoremstyle{remark}

\title{Throw One's Cake --- and Eat It Too}
\author{Orit Arzi, Yonatan Aumann \& Yair Dombb \\ Bar-Ilan University\thanks{Department of Computer Science, Bar-Ilan University, Ramat Gan 52900, Israel.  Email addresses: \mbox{oritarzi1@gmail.com}, aumann@cs.biu.ac.il, yair.biu@gmail.com}}
\date{}

\begin{document}

\maketitle

``Envy, desire, and the pursuit of honor drive a person from the world.''
\begin{flushright}
\it{Chapters of Our Fathers 4:27} \qquad\qquad\qquad
\end{flushright}

\begin{abstract}
We consider the problem of fairly dividing a heterogeneous cake between a number of players with different tastes. In this setting, it is known that fairness requirements may result in a suboptimal division from the social welfare standpoint. Here, we show that in some cases, discarding  some of the cake and fairly dividing only the remainder may be socially preferable to any fair division of the entire cake. We study this phenomenon, providing asymptotically-tight bounds on the social improvement achievable by such discarding.
\end{abstract}

\section{Introduction}

{\em Cake cutting} is a standard metaphor used for modeling the problem of fair division of goods among multiple players. ``Fairness'' can be defined in several different ways, with {\em envy-freeness} being one of the more prominent ones.  A division is {\em envy-free} if no player 
prefers getting a piece given to someone else.

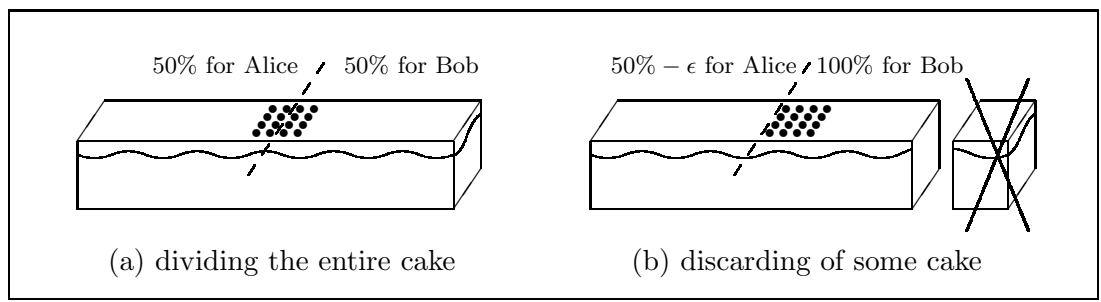
\begin{figure}[bht]
\begin{center}
\framebox[1.1\width]{

\setlength{\unitlength}{9mm}
\begin{picture}(14.2,4)

\newsavebox{\icing}
\savebox{\icing}{
    \qbezier(0,0.1)(0.25,0)(0.5,0.1) \qbezier(0.5,0.1)(0.75,0.2)(1,0.1)
}

\newsavebox{\brokenline}
\savebox{\brokenline}{
    \qbezier(0,0)(0.1,0.15)(0.1,0.15)
    \qbezier(0.2,0.3)(0.3,0.45)(0.3,0.45)
    \qbezier(0.4,0.6)(0.5,0.75)(0.5,0.75)
    \qbezier(0.6,0.9)(0.7,1.05)(0.7,1.05)
    \qbezier(0.8,1.2)(0.9,1.35)(0.9,1.35)
    \qbezier(1.0,1.5)(1.1,1.65)(1.1,1.65)
}

\newsavebox{\leftend}
\savebox{\leftend}{
    \put(0,0){\line(0,1){1}} \put(0,1){\line(2,3){0.4}}
}

\newsavebox{\rightend}
\savebox{\rightend}{
    \put(0,0){\usebox{\leftend}}
    \put(0,0){\line(2,3){0.4}} \put(0.4,0.6){\line(0,1){1}}
}

\newsavebox{\cherries}
\savebox{\cherries}{
\put(0,0){\circle{0.2}} \put(0.12,0.18){\circle{0.2}}
\qbezier(0,0.04)(-0.08,0.38)(0.14,0.61) \qbezier(0.12,0.22)(-0.04,0.44)(0.14,0.61)
}

\newsavebox{\oldcherries}
\savebox{\oldcherries}{
\put(0,0){\circle{0.2}} \put(0.3,0){\circle{0.2}}
\qbezier(0,0.04)(0.09,0.44)(0.31,0.53) \qbezier(0.3,.04)(0.09,0.43)(0.31,0.53)
}

\put(0,1.2){\usebox{\leftend}}
\put(0,1.2){\line(1,0){5.5}}
\put(0,2.2){\line(1,0){5.5}}
\put(0.4,2.8){\line(1,0){5.5}}
\put(5.5,1.2){\usebox{\rightend}}

\put(0,1.9){\usebox{\icing}} \put(1,1.9){\usebox{\icing}} \put(2,1.9){\usebox{\icing}}\put(3,1.9){\usebox{\icing}} \put(4,1.9){\usebox{\icing}}
\qbezier(5,2.0)(5.25,1.9)(5.5,2.0)
\qbezier(5.5,2.0)(5.6,2.0)(5.7,2.3) \qbezier(5.7,2.3)(5.8,2.55)(5.9,2.6)

\multiput(2.62,2.32)(0.08,0.12){4}{\circle*{0.1}}
\multiput(2.82,2.32)(0.08,0.12){4}{\circle*{0.1}}
\multiput(3.02,2.32)(0.08,0.12){4}{\circle*{0.1}}
\multiput(3.22,2.32)(0.08,0.12){4}{\circle*{0.1}}

\put(2.5,1.7){\usebox{\brokenline}}

\put(1.1,3.2){\footnotesize $50\%$ for Alice}
\put(3.9,3.2){\footnotesize $50\%$ for Bob}

\put(0.45,0.3){(a) dividing the entire cake}

\put(7.5,1.2){\usebox{\leftend}}
\put(7.5,1.2){\line(1,0){4.7}}
\put(7.5,2.2){\line(1,0){4.7}}
\put(7.9,2.8){\line(1,0){4.7}}
\put(12.2,1.2){\usebox{\rightend}}

\put(7.5,1.9){\usebox{\icing}} \put(8.5,1.9){\usebox{\icing}}
\put(9.5,1.9){\usebox{\icing}} \put(10.5,1.9){\usebox{\icing}}
\qbezier(11.5,2.0)(11.75,1.9)(12,2.0) \qbezier(12,2.0)(12.1,2.05)(12.2,2.08)

\multiput(10.12,2.32)(0.08,0.12){4}{\circle*{0.1}}
\multiput(10.32,2.32)(0.08,0.12){4}{\circle*{0.1}}
\multiput(10.52,2.32)(0.08,0.12){4}{\circle*{0.1}}
\multiput(10.72,2.32)(0.08,0.12){4}{\circle*{0.1}}

\put(9.6,1.7){\usebox{\brokenline}}

\put(12.8,1.2){\usebox{\leftend}}
\put(12.8,1.2){\line(1,0){0.8}}
\put(12.8,2.2){\line(1,0){0.8}}
\put(13.2,2.8){\line(1,0){0.8}}
\put(13.6,1.2){\usebox{\rightend}}

\qbezier(12.8,2.08)(12.85,2.1)(13.1,2.0) \qbezier(13.1,2.0)(13.35,1.9)(13.6,2.0)
\qbezier(13.6,2.0)(13.7,2.05)(13.8,2.3) \qbezier(13.8,2.3)(13.9,2.55)(14,2.6)

\linethickness{0.25mm}
\qbezier(13.0,0.885)(13.9,3.115)(13.9,3.115)
\qbezier(13.0,3.115)(13.9,0.885)(13.9,0.885)

\put(7.8,3.2){\footnotesize $50\%-\epsilon$ for Alice}
\put(10.8,3.2){\footnotesize $100\%$ for Bob}

\put(8.1,0.3){(b) discarding of some cake}

\end{picture}}
\caption{Discarding a small part of the cake allows for a better envy-free division between Alice and Bob.}\label{fig:cake-intro}
\end{center}
\end{figure}

Consider the rectangular cake depicted in Figure~\ref{fig:cake-intro}(a). It is a chocolate cake sprinkled with candies right along the middle.  Suppose you have two kids: Alice and Bob.  Alice likes the base of the cake, but is indifferent to the candies; Bob is the opposite: he cares only for the candies. It is easy to see that if each of the children must get one consecutive piece of the cake, then splitting the cake along the middle is the only possible envy-free division. Any other split would result in one child getting less than 50\% (by his or her valuation) and envying the other. But this division is rather wasteful: if Bob could only get a small additional fraction from Alice (small in her view), he would be doubly as happy. Is there any possible way to make this happen --- without introducing envy? 

Interestingly, the answer is in the affirmative.  By discarding a small piece from the right-end of the cake, we can now place the cut to the left of candies, giving the right piece to Bob and the left to Alice (see Figure~\ref{fig:cake-intro}(b)).  Alice would no longer envy Bob, as he gets the same amount of the cake as she does. The overall happiness level would substantially increase: Bob is doubly happy, and Alice is only $\epsilon$ less happy.

The above is a particular example of what we call the {\em dumping paradox} --- the  phenomenon in which one can increase the social welfare of envy-free divisions by discarding (={\em dumping}) some of the cake. The above provides an example of a {\em utilitarian $1.5-\epsilon$ dumping paradox}  (i.e.~the utilitarian welfare, defined as the sum of individual utilities, increases by a factor of $1.5-\epsilon$).  In this paper we analyze the dumping paradox: under what circumstances may it arise? what social welfare can it improve? and by how much?
Interestingly, we show that at times much can be gained by such discarding of some of the cake. 
We show:
\begin{itemize}

\item With regards to utilitarian welfare, the dumping paradox with $n$ players can be as high as $\Theta(\sqrt{n})$; i.e.\ there are cases where discarding some of the cake allows for an envy-free division that is $\Theta(\sqrt{n})$ better (from the utilitarian standpoint) than any envy-free division of the entire cake. This bound is asymptotically tight. 

\item With regards to egalitarian welfare, the dumping paradox with $n$ players can be as high as $\frac{n}{3}$. Egalitarian welfare is defined as the utility obtained by the least-happy player.  In particular, we show a case where discarding some cake allows us to improve from an allocation in which at least one player gets no more than $1/n$ to an allocation in which everybody gets at least $\approx \frac{1}{3}$ (!). 
Our construction almost matches the upper bound of $\frac{n}{2}$ following from~\cite{AD10}; for $n\leq 4$ we show that the bound of $\frac{n}{2}$ can actually be obtained. 

\item With regards to Pareto efficiency, there are instances in which discarding some cake allows for an envy-free division that Pareto dominates every envy-free division of the entire cake. We show that by discarding even one piece of the cake it may be possible to \emph{double} the utility of all but two players without harming these remaining two players. 

\end{itemize}
All of our results are for divisions that require that each player get one consecutive piece of the cake.  For divisions that allow players to get arbitrarily many pieces of the cake we show that no dumping paradox is possible.

\paragraph{Related work.} 
The problem of fair division has been studied in many different fields and settings. Modern mathematical treatment of fair division via the cake cutting abstraction started in the 1940s~\cite{Ste49}. Since then, many works presented algorithms or protocols for fair division~\cite{Str80,EP84,BT95,CLPP10}, as well as theorems establishing the existence of fair divisions (under different interpretations of fairness) in different settings~\cite{DS61,Str80}. Starting from the mid 1990s, several books appeared on the subject~\cite{BT96,RW98,Mou04}, and much attention was given to the question of finding bounds on the number of steps required for dividing a cake fairly~\cite{MIBK03,SW03,EP06,Pro09}.

A more recent work by Caragiannis et al.~\cite{CKKK09} added the issue of social welfare into the framework of cake cutting. In particular, Caragiannis et al.~aimed at showing bounds on the loss of social welfare caused by fairness requirements, by defining and analyzing the Price of Fairness (defined for different fairness criteria). The work in~\cite{CKKK09} considered fair division of divisible and indivisible goods, as well as divisible and indivisible chores; for each of these settings, it has provided bounds on the highest possible degradation in utilitarian welfare caused by three prominent fairness requirements --- proportionality, envy-freeness, and equitability. 
Following this line of work, a recent work by a subset of the authors here~\cite{AD10} analyzed the utilitarian and egalitarian Price of Fairness in the setting of cake cutting where each piece is required to be a single connected interval. (This is in contrast to the work of~\cite{CKKK09} that allowed a piece in the division to be comprised of any union of intervals.)

Finally, the concept of partial divisions of a cake (in which not all the cake is allotted to the players) has also been considered in~\cite{CLPP10} and very recently in~\cite{CLP11}. Interestingly, in each of these works, discarding of some of the cake serves a different purpose. In~\cite{CLPP10}, the authors present a proportional, envy-free and truthful cake cutting algorithm for players with valuation functions of a restricted form. In that work, disposing of some of the cake is what ensures that the players have no incentive to lie to the protocol. In~\cite{CLP11}, a restricted case of non-additive valuation functions is considered, with one of the results being an approximately-proportional, envy-free protocol for two players. In that protocol, some cake is discarded in order to guarantee envy-freeness. Here, we show that leaving some cake unallocated can also increase social welfare.

\section{Definitions and Preliminaries}

As customary, we assume a 1-dimensional cake that is represented by the interval $[0,1]$. We denote the set of players by $[n]$ (where $[n] = \{1,\ldots,n\}$), and assume that each player $i$ has a nonatomic (additive) measure $v_i$ mapping each interval of the cake to its value for player $i$, and having $v_i([0,1])=1$. Let $x$ be some division of the cake  between the players; we denote the value player $i$ assigns to player $j$'s piece in $x$ by $u_i(x,j)$. We say that a division $x$ is \emph{complete} if it leaves no cake unallocated; otherwise, we say that the division is \emph{partial}.

\begin{definition}
We say that a cake instance with $n$ players exhibits an \emph{$\alpha$-dumping paradox} (with $\alpha>1$ and with respect to some social welfare function $w(\cdot)$) if there exists a partial division $y$ such that
\begin{enumerate}
\item $y$ is envy-free; i.e.~$u_i(y,i) \geq u_i(y,j)$ for all $i,j\in[n]$, and

\item for every envy-free complete division $x$, $w(y) \geq \alpha\cdot w(x)$.
\end{enumerate}
\end{definition}

In this work, we consider two prominent social welfare functions: utilitarian and egalitarian. The utilitarian welfare of a division $x$ is the sum of the players' utilities; formally, we write $u(x) = \sum_{i\in[n]}{u_i(x,i)}$. The egalitarian welfare of a division $x$ is the utility of the worst-off player, i.e.~$eg(x) = \min_{i\in[n]}{u_i(x,i)}$.

From this point forward, we restrict the discussion to divisions in which every player gets a single connected interval of the cake. The first reason for this restriction is that giving the players such connected pieces seem more ``natural'', and is in many scenarios more desirable than giving pieces composed of unions of intervals. The second reason is captured by the following simple result:

\begin{proposition}
If players are allowed to get non-connected pieces (that are composed of unions of intervals), there can be no utilitarian or egalitarian dumping paradox. In addition, in this setting no envy-free partial division can Pareto dominate all envy-free complete divisions.
\end{proposition}

\begin{proof}
We prove for utilitarian welfare; the proof for egalitarian welfare and Pareto domination is analogous. Suppose that we allow such non-connected divisions, and assume that there is a utilitarian $\alpha$-dumping paradox, with $\alpha>1$. Then there exists an envy-free partial division $y$ such that for every envy-free complete division $x$, $\sum_{i\in[n]}{u_i(y,i)} \geq \alpha\cdot\sum_{i\in[n]}{u_i(x,i)}$. Let $U\subseteq [0,1]$ be the part of the cake that was not allocated to the players in $y$; it is known (e.g.~\cite{DS61}) that $U$ itself has a complete envy-free division $y'$. Note that giving each player $i$ her part from $y'$ in addition to her original piece from $y$ yields again an envy-free division; call this division $z$. Clearly, $z$ is a complete division of $[0,1]$ having $u_i(z,i) \geq u_i(y,i)$ for all $i\in[n]$. It follows that for every envy-free division $x$
\begin{equation*}
\sum_{i\in[n]}{u_i(z,i)} \geq \sum_{i\in[n]}{u_i(y,i)} \geq \alpha\cdot\sum_{i\in[n]}{u_i(x,i)} > \sum_{i\in[n]}{u_i(x,i)} \;;
\end{equation*}
a contradiction.
\end{proof}

We thus formally define a connected division of a cake to $n$ players simply as a sequence of $n$ non-intersecting open intervals\footnote{Since we assume that the valuation functions of all players are nonatomic, open and closed intervals always have the same value.}; the first interval is given to the first player, the second to the second player, etc. We will say that such a division is complete if the union of these intervals (including their endpoints) equals the entire cake; otherwise, we will say that the division is partial. Note that a partial division may leave several disjoint intervals unallocated.

Finally, we give the definition of the Price of Envy-Freeness, first defined in~\cite{CKKK09}, which aims to measure the highest degradation in social welfare that may be necessary to achieve envy-freeness.

\begin{definition}
Let $I$ be a cake instance, $X$ the set of all complete divisions of $I$, and $X_{EF}$ the set of all complete envy-free divisions of $I$.
The \emph{Price of Envy-Freeness} of the cake instance $I$, with respect to a social welfare function $w(\cdot)$, is defined as the ratio
\begin{equation*}
\frac{\max_{x\in X}{w(x)}}{\max_{y\in X_{EF}}{w(y)}} \;.
\end{equation*}
\end{definition}

We now show a connection between the dumping paradox of a cake instance and the Price of Envy-Freeness for the same instance and welfare function.

\begin{proposition}\label{pro:de-pof}
The utilitarian (resp.~egalitarian) dumping paradox is bounded from above by the utilitarian (resp.~egalitarian) Price of Envy-Freeness.
\end{proposition}

\begin{proof}
We again prove only for utilitarian welfare. Assume, by contradiction, that there exists a cake cutting instance with $n$ players and with utilitarian dumping paradox of $\beta$, while the utilitarian Price of Envy-Freeness for these $n$ players is $\alpha < \beta$. Then there exists a partial division $y$ such that for every envy-free complete division $x$, $\sum_{i\in[n]}{u_i(y,i)} \geq \beta\cdot\sum_{i\in[n]}{u_i(x,i)}$. Note that every inclusion-maximal unalloted interval is adjacent to at least one interval that is given to a player. Therefore, consider the complete division $z$ which allocates each player her interval as in $y$, and in addition adds the previously-unalloted intervals to the piece of one of the adjacent players (chosen arbitrarily). This is clearly a (not necessarily envy-free) complete division in which $u_i(z,i) \geq u_i(y,i)$ for every $i\in[n]$. We get that for every envy-free division $x$
\begin{equation*}
\sum_{i\in[n]}{u_i(z,i)} \geq \sum_{i\in[n]}{u_i(y,i)} \geq \beta\cdot\sum_{i\in[n]}{u_i(x,i)} > \alpha\cdot\sum_{i\in[n]}{u_i(x,i)} \;,
\end{equation*}
contradicting our bound on the utilitarian Price of Envy-Freeness.
\end{proof}

\section{Utilitarian Welfare}

\begin{theorem}\label{thm:duut}
The utilitarian dumping paradox with $n$ players may be as high as $\Theta(\sqrt{n})$, and this bound is asymptotically tight.
\end{theorem}

We will show that for every $k,t\in\mathbb{N}$, there exists a cake cutting instance with $n = 2k(3t-2)$ players in which throwing away $(t-1)k$ intervals of the cake can improve the utilitarian welfare of the best envy-free division by a factor of $\frac{ktn + 2tn}{nk + 12t(t-1)k + tn} > \frac{kt+2t}{k+3t}$ (note that $12(t-1)k < 2n$). In particular, choosing $t = \Theta(k)$ yields an improvement of $\Theta(\sqrt{n})$. The matching upper bound follows from Proposition~\ref{pro:de-pof}, combined with with Theorem 1 of~\cite{AD10}, which shows an upper bound of $\frac{\sqrt{n}}{2} + 1 - o(1)$ on the Price of Envy-Freeness.

To prove the lower bound we construct a cake with three parts: the ``common'' part, the ``high-values'' part, and the ``compensation'' part. Furthermore, each of the latter two parts is itself divided into $k$ identical subparts. Thus, in order to clarify the presentation, we first illustrate the key structure and reason about its properties. We then explain how this structure is used to create the full construction.

Let us start with a subset of $3t-2$ players, comprised of $2t-2$ ``Type A'' players, $t-1$ ``Type B'' players, and one ``chosen'' player $C$. For $0\leq i\leq t-2$ we will have the players $3i + 1$ and $3i + 2$ be of Type A, and the player $3i+3$ be of Type B; we will say that players $3i+1$ and $3i+2$ are ``neighbors''. In the ``high values'' part of the cake, the chosen player $C$ has $t$ intervals she desires, each of them being of value $\frac{1}{t}$ to her. Between every two consecutive intervals desired by $C$ there are two more intervals, desired by a pair of neighbors of Type A. Each neighbor desires one of these intervals, and considers it to be worth  $\frac{4}{n}$ (as can be seen on the left-hand side of Figure~\ref{fig:prefut}). In addition, each pair of neighbors desires two more intervals, located in the ``compensation'' part of the cake. Specifically, for every two neighbors $3i+1$ and $3i+2$, we have two intervals desired by $3i+1$ followed by two intervals desired by $3i+2$; each of these intervals is worth $\frac{2}{n}$ to the corresponding player. In between these four intervals, there are three intervals desired by the player $3i+3$ of Type B: the first and third ones have each a value of $\frac{3}{2n}$ to that player, and the second has value of $\frac{1}{n}$. The reader is again referred to Figure~\ref{fig:prefut} for a graphical representation; note that while the preferences of the chosen player $C$ are completely described (as the union of all her desired intervals here sums up to $1$), this is not so for the other players.

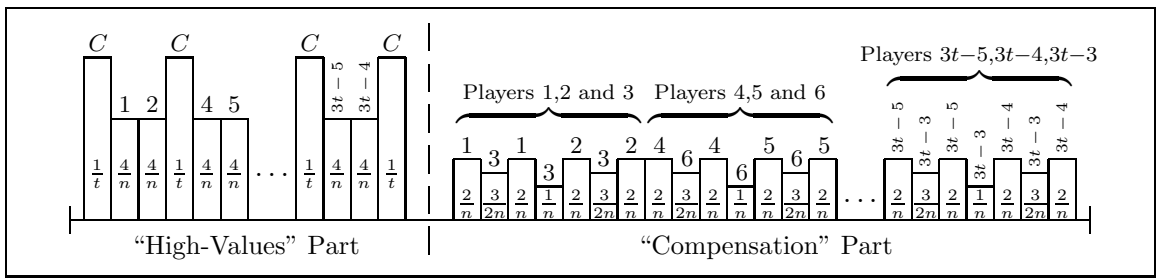
\begin{figure}[bth]
\begin{center}
\framebox[1.1\width]{

\setlength{\unitlength}{9mm}
\begin{picture}(15,3.7)

\put(-0.1,0.7){\line(1,0){14.9}}
\put(-0.1,0.5){\line(0,1){0.4}}
\put(14.8,0.5){\line(0,1){0.4}}

\newsavebox{\rectT}
\savebox{\rectT}(2.4,0.4)[bl]{
	\put(0,0){\line(0,1){2.4}} \put(0,2.4){\line(1,0){0.4}}
	\put(0.4,0){\line(0,1){2.4}}
}
\newsavebox{\rectK}
  \savebox{\rectK}(0.4,0.5)[bl]{
	\put(0,0){\line(0,1){0.5}} \put(0,0.5){\line(1,0){0.4}}
	\put(0.4,0){\line(0,1){0.5}}
}
	
\global\newsavebox{\rectMk}
  \savebox{\rectMk}(0.4,0.7)[bl]{
	\put(0,0){\line(0,1){0.7}} \put(0,0.7){\line(1,0){0.4}}
	\put(0.4,0){\line(0,1){0.7}}
}

\global\newsavebox{\rectDk}
  \savebox{\rectDk}(0.4,0.9)[bl]{
	\put(0,0){\line(0,1){0.9}} \put(0,0.9){\line(1,0){0.4}}
	\put(0.4,0){\line(0,1){0.9}}
}

\global\newsavebox{\rectQk}
  \savebox{\rectQk}(0.4,1.5)[bl]{
	\put(0,0){\line(0,1){1.5}} \put(0,1.5){\line(1,0){0.4}}
	\put(0.4,0){\line(0,1){1.5}}
}

\put(0.1,0.7){\usebox{\rectT}}     \put(0.18,1.3){\footnotesize $\frac{1}{t}$}   \put(0.17,3.2){\footnotesize $C$} 
\put(0.5,0.7){\usebox{\rectQk}}  \put(0.56,1.3){\footnotesize $\frac{4}{n}$}   \put(0.61,2.3){\footnotesize $1$} 
\put(0.9,0.7){\usebox{\rectQk}}  \put(0.96,1.3){\footnotesize $\frac{4}{n}$}   \put(1.01,2.3){\footnotesize $2$} 
\put(1.3,0.7){\usebox{\rectT}}     \put(1.38,1.3){\footnotesize $\frac{1}{t}$}   \put(1.37,3.2){\footnotesize $C$} 
\put(1.7,0.7){\usebox{\rectQk}}  \put(1.76,1.3){\footnotesize $\frac{4}{n}$}   \put(1.81,2.3){\footnotesize $4$} 
\put(2.1,0.7){\usebox{\rectQk}}  \put(2.16,1.3){\footnotesize $\frac{4}{n}$}   \put(2.21,2.3){\footnotesize $5$}

\put(2.6, 1.36){\dots}

\put(3.2,0.7){\usebox{\rectT}}     \put(3.28,1.3){\footnotesize $\frac{1}{t}$}   \put(3.27,3.2){\footnotesize $C$} 
\put(3.6,0.7){\usebox{\rectQk}}  \put(3.66,1.3){\footnotesize $\frac{4}{n}$}   
\put(4.0,0.7){\usebox{\rectQk}}  \put(4.06,1.3){\footnotesize $\frac{4}{n}$}   
\put(4.4,0.7){\usebox{\rectT}}     \put(4.48,1.3){\footnotesize $\frac{1}{t}$}   \put(4.47,3.2){\footnotesize $C$} 

\multiput(5.15, 0.2)(0,0.6){6}{\line(0,1){0.4}}

\put(0.8, 0.2){\small ``High-Values'' Part}
\put(8.2, 0.2){\small ``Compensation'' Part}

\put(5.55,2.1){$\overbrace{\qquad\qquad\quad\quad\;}^\text{Players 1,2 and 3}$}

\put(5.5,0.7){\usebox{\rectDk}}     \put(5.56,0.86){\footnotesize $\frac{2}{n}$}   \put(5.61,1.7){\footnotesize $1$}
\put(5.9,0.7){\usebox{\rectMk}}     \put(5.9,0.86){\footnotesize $\frac{3}{2n}$}   \put(6.01,1.5){\footnotesize $3$}
\put(6.3,0.7){\usebox{\rectDk}}     \put(6.36,0.86){\footnotesize $\frac{2}{n}$}   \put(6.41,1.7){\footnotesize $1$}
\put(6.7,0.7){\usebox{\rectK}}     \put(6.76,0.86){\footnotesize $\frac{1}{n}$}   \put(6.81,1.3){\footnotesize $3$}
\put(7.1,0.7){\usebox{\rectDk}}     \put(7.16,0.86){\footnotesize $\frac{2}{n}$}   \put(7.21,1.7){\footnotesize $2$}
\put(7.5,0.7){\usebox{\rectMk}}     \put(7.5,0.86){\footnotesize $\frac{3}{2n}$}   \put(7.61,1.5){\footnotesize $3$}
\put(7.9,0.7){\usebox{\rectDk}}     \put(7.96,0.86){\footnotesize $\frac{2}{n}$}   \put(8.01,1.7){\footnotesize $2$}

\put(8.35,2.1){$\overbrace{\qquad\qquad\quad\quad\;}^\text{Players 4,5 and 6}$}

\put(8.3,0.7){\usebox{\rectDk}}     \put(8.36,0.86){\footnotesize $\frac{2}{n}$}   \put(8.43,1.7){\footnotesize $4$}
\put(8.7,0.7){\usebox{\rectMk}}     \put(8.7,0.86){\footnotesize $\frac{3}{2n}$}   \put(8.83,1.5){\footnotesize $6$}
\put(9.1,0.7){\usebox{\rectDk}}     \put(9.16,0.86){\footnotesize $\frac{2}{n}$}   \put(9.23,1.7){\footnotesize $4$}
\put(9.5,0.7){\usebox{\rectK}}     \put(9.56,0.86){\footnotesize $\frac{1}{n}$}   \put(9.63,1.3){\footnotesize $6$}
\put(9.9,0.7){\usebox{\rectDk}}     \put(9.96,0.86){\footnotesize $\frac{2}{n}$}   \put(10.03,1.7){\footnotesize $5$}
\put(10.3,0.7){\usebox{\rectMk}}     \put(10.3,0.86){\footnotesize $\frac{3}{2n}$}   \put(10.43,1.5){\footnotesize $6$}
\put(10.7,0.7){\usebox{\rectDk}}     \put(10.76,0.86){\footnotesize $\frac{2}{n}$}   \put(10.83,1.7){\footnotesize $5$}

\put(11.2,0.97){\dots}

\put(11.45,2.65){$\overbrace{\qquad\qquad\quad\quad\;}^{\text{Players }3t-5,3t-4,3t-3}$}

\put(11.8,0.7){\usebox{\rectDk}}     \put(11.86,0.86){\footnotesize $\frac{2}{n}$}   
\put(12.2,0.7){\usebox{\rectMk}}     \put(12.2,0.86){\footnotesize $\frac{3}{2n}$}   
\put(12.6,0.7){\usebox{\rectDk}}     \put(12.66,0.86){\footnotesize $\frac{2}{n}$}   
\put(13.0,0.7){\usebox{\rectK}}     \put(13.06,0.86){\footnotesize $\frac{1}{n}$}   
\put(13.4,0.7){\usebox{\rectDk}}     \put(13.46,0.86){\footnotesize $\frac{2}{n}$}   
\put(13.8,0.7){\usebox{\rectMk}}     \put(13.8,0.86){\footnotesize $\frac{3}{2n}$}   
\put(14.2,0.7){\usebox{\rectDk}}     \put(14.26,0.86){\footnotesize $\frac{2}{n}$}   

\put(3.73,2.29){\begin{sideways}\tiny $3t-5$ \end{sideways}}
\put(4.13,2.29){\begin{sideways}\tiny $3t-4$ \end{sideways}}

\put(11.92,1.69){\begin{sideways}\tiny $3t-5$ \end{sideways}}
\put(12.32,1.49){\begin{sideways}\tiny $3t-3$ \end{sideways}}
\put(12.72,1.69){\begin{sideways}\tiny $3t-5$ \end{sideways}}
\put(13.12,1.29){\begin{sideways}\tiny $3t-3$ \end{sideways}}
\put(13.52,1.69){\begin{sideways}\tiny $3t-4$ \end{sideways}}
\put(13.92,1.49){\begin{sideways}\tiny $3t-3$ \end{sideways}}
\put(14.32,1.69){\begin{sideways}\tiny $3t-4$ \end{sideways}}

\end{picture}}
\end{center}\caption{The (incomplete) preferences of one set of players. The number above each column denotes which player has that valuation.}
\label{fig:prefut}
\end{figure}

We can now reason about the possible envy-free divisions for these (incompletely-described) preferences.

\begin{lemma}\label{lem:u-ef-chosen}
Assume a set of players with preferences as above. Suppose in addition that in any envy-free division the following properties hold:
\begin{enumerate}
\renewcommand{\theenumi}{(P\arabic{enumi})}
\renewcommand{\labelenumi}{\theenumi}
\setlength\itemindent{10px}

\item any piece of value $\frac{4}{n}$ or more that can only be given to some Type A player either intersects her (single) desired interval from the ``high-values'' part, or contains both of her desired intervals from the ``compensation'' part, and \label{asmp1}

\item it is impossible for any Type B player to get a piece of value $\geq \frac{2}{n}$ which is not completely contained in the part of the cake described above. \label{asmp2}
\end{enumerate}
Then no envy free division gives the chosen player $C$ a piece of value  $> \frac{1}{t}$.
\end{lemma}

\begin{proof}
Suppose that $C$ does get such a piece. It must be that this piece intersects at least two of $C$'s desired intervals; in other words, there are two neighbors of Type A such that $C$ completely devours a desired interval of each of them. Let these players be $3i+1$ and $3i+2$; these players consider $C$'s piece as worth at least $\frac{4}{n}$ and thus must each get a piece of at least this value to avoid envy. By the property~\ref{asmp1}, the only way to do that is to give each of them their two desired intervals from the compensation part of the cake. Recall that each of the two players has two desired intervals in the compensation part: denote them (from left to right) $A_1,A_2,A_3$ and $A_4$. In between those pieces, there are three intervals which we will denote by $B_1,B_2$ and $B_3$, desired by the Type B player $3i+3$. In order to give each of these Type A players a piece of value $\frac{4}{n}$, we must give player $3i+1$ a contiguous piece containing $A_1,B_1$ and $A_2$, and player $3i+2$ a piece containing $A_3,B_3$ and $A_4$. Each of these pieces is worth at least $\frac{3}{2n}$ to player $3i+3$ who thus cannot be satisfied with the piece $B_2$ (worth to her only $\frac{1}{n}$), and must therefore get her share from another part of the cake. Since no other players have any value for the interval $B_2$, it must be shared between players $3i+1$ and $3i+2$ whose pieces are the closest to it. At least one of these players will get at least half of $B_2$, and the piece of this player will be worth at least $\frac{2}{n}$ to player $3i+3$; by the property~\ref{asmp2}, this will cause envy.
\end{proof}

We can now fully describe our construction. We will have $k$ sets of players, each of them identical to the set described above. This sums up to $k$ chosen players, $k(2t-2)$ Type A players, and $k(t-1)$ Type B players, totaling in $k(3t-2) = \frac{n}{2}$ players; the other half of the players will be called ``the common players'', and will all have the same preferences. The leftmost part of the cake will be the ``common'' part; this part is worth $1$ to all of the common players, $1-\frac{8}{n}$ to the Type A players, and $1-\frac{4}{n}$ to the Type B players. In the middle, we will have the ``high-values'' part, which will be composed of $k$ copies of the high-values part presented above, one for each set of players. Finally, the rightmost part of the cake will be the ``compensation'' part, which will again be composed of $k$ identical copies of the compensation part presented above, in the same order of sets of players as the high-values part. The reader is referred to Figure~\ref{fig:preut2}, which illustrates the structure of the full construction.

\begin{figure}[bth]
\begin{center}
\framebox[1.1\width]{

\setlength{\unitlength}{5.5mm}
\begin{picture}(24.5,5)

\put(-0.2,1.2){\line(1,0){24.4}}
\put(-0.2,1){\line(0,1){0.4}}
\put(24.2,1){\line(0,1){0.4}}

\newsavebox{\HV}
 \savebox{\HV}(3.6,1.8)[bl]{
	\put(0,0){\line(0,1){1.8}} \put(0.3,0){\line(0,1){1.8}} 
	\put(0.9,0){\line(0,1){1.8}} \put(1.2,0){\line(0,1){1.8}} 
	\put(2.4,0){\line(0,1){1.8}} \put(2.7,0){\line(0,1){1.8}} 
	\put(3.3,0){\line(0,1){1.8}} \put(3.6,0){\line(0,1){1.8}} 
	
	\put(0.6,0){\line(0,1){1.125}}
	\put(1.5,0){\line(0,1){1.125}} \put(1.8,0){\line(0,1){1.125}}
	\put(3.0,0){\line(0,1){1.125}}
		
	\put(0,1.8){\line(1,0){0.3}}	\put(0.9,1.8){\line(1,0){0.3}}
	\put(2.4,1.8){\line(1,0){0.3}} \put(3.3,1.8){\line(1,0){0.3}} 
	\put(0.3,1.125){\line(1,0){0.6}} \put(1.2,1.125){\line(1,0){0.6}}
	\put(2.7,1.125){\line(1,0){0.6}} 
	
	\put(1.82,0.5){...}
}

\newsavebox{\CompB}
 \savebox{\CompB}(2.1,0.675)[bl]{
 	\multiput(0,0)(0.3,0){8}{\line(0,1){0.675}} 
 	\multiput(0,0.675)(0.6,0){4}{\line(1,0){0.3}} 
 	\multiput(0.3,0.525)(1.2,0){2}{\line(1,0){0.3}} 
 	\put(0.9,0.375){\line(1,0){0.3}} 	
}

\newsavebox{\Comp}
 \savebox{\Comp}(4.8,0.675)[bl]{
 	\put(0,0){\usebox{\CompB}}
 	\put(2.12,0.2){...}
 	\put(2.7,0){\usebox{\CompB}}
}

\put(0.05,0.3){\footnotesize ``Common'' Part}

\put(0,1.2){\line(0,1){0.65}} \put(4.4,1.2){\line(0,1){0.65}}
\put(0,1.85){\line(1,0){4.4}}
\multiput(0.3,1.2)(0.65,0){6}{\line(1,1){0.65}}

\multiput(4.6,0.85)(0,0.2){17}{\line(0,1){0.1}}
\put(6.3,0.3){\footnotesize ``High-Values'' Part}

\put(4.8,1.2){\usebox{\HV}} \put(4.81,3.4){$\overbrace{\qquad\qquad\ \ \;\,}^\text{First Set}$}
\put(8.52,1.9){\dots}
\put(9.41,3.4){$\overbrace{\qquad\qquad\ \ \;\,}^{k\text{-th Set}}$}
\put(9.4,1.2){\usebox{\HV}}

\multiput(13.2,0.85)(0,0.2){17}{\line(0,1){0.1}}
\put(15.8,0.3){\footnotesize ``Compensation'' Part}

\put(13.43,2.2){$\overbrace{\qquad\qquad\qquad\ \;\,}^\text{First Set}$}
\put(13.4,1.2){\usebox{\Comp}}
\put(18.32,1.5){\dots}
\put(19.23,2.2){$\overbrace{\qquad\qquad\qquad\ \;\,}^{k\text{-th Set}}$}
\put(19.2,1.2){\usebox{\Comp}}

\end{picture}} \caption{Preferences of all players when $n=2(3t-2)k$.  }\label{fig:preut2}
\end{center}
\end{figure}
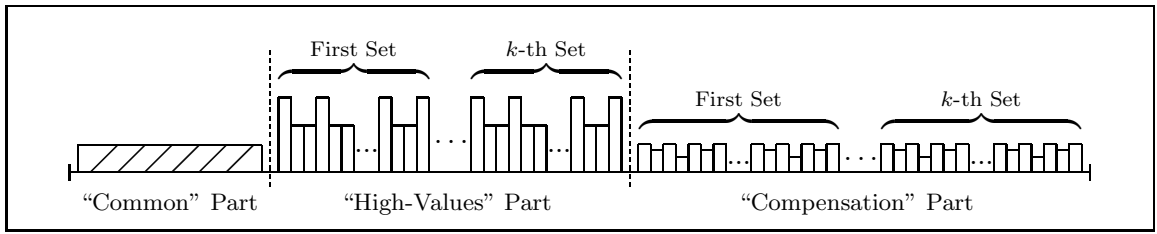

\begin{lemma}\label{lem:u-condition}
The properties~\ref{asmp1} and~\ref{asmp2} of Lemma~\ref{lem:u-ef-chosen} hold in our construction for all the Type A and Type B players in all of the sets.
\end{lemma}

\begin{proof}
Note that no player can get a piece which intersects both the common part and the compensation part (as such a connected piece contains the entire ``high values'' part of the cake). Thus,~\ref{asmp1} and~\ref{asmp2} follow by showing that in an envy-free division no Type A or Type B player can get a piece of value $\frac{2}{n}$ or more from common part alone. 

Suppose that we do give some Type A or Type B player a piece of value at least $\frac{2}{n}$ from the common part; such a piece must be of physical size of more than $\frac{2}{n}$-th of the total size of the common part. However, in order to avoid envy, we must then give each of the common players a piece of at least this size from the common part. This forces us to give the players in total at least $(\frac{n}{2}+1)\cdot\frac{2}{n} > 1$ of the size of the common part; i.e.~we need to give away more than $100\%$ of this part.
\end{proof}

We can now prove a bound on the utilitarian welfare of any envy-free division.

\begin{lemma}\label{u-ef-welfare}
Any envy-free division of the cake described above has utilitarian welfare of at most
$\big(\frac{1}{t} + \frac{12(t - 1)}{n}\big)k + 1$.
\end{lemma}

\begin{proof}
Consider the following division. We divide the common part equally between all the common players; this gives each common player value of $\frac{2}{n}$, and contributes a total of $1$ to the utilitarian welfare. We next divide the high-values part: we give the first desired interval of each chosen player to that player; these players thus contribute a total of $\frac{k}{t}$ to the welfare. We also give each Type A player her (single) desired interval from this part. (This leaves us with $(t-1)k$ unallocated intervals of this part that are desired only by the chosen players; we add each such interval to the piece of one of the players whose pieces are closest.) The collective contribution of the Type A players to the welfare is $2k(t-1)\cdot\frac{4}{n}$. Finally, we divide the compensation part between the Type B players: we give each such player an piece containing all her desired intervals from the compensation part, adding the remaining intervals to any of the closest pieces. Thus, the Type B players collectively contribute $k(t-1)\cdot\frac{4}{n}$ to the welfare; adding everything up, we get a division with utilitarian welfare of $\big(\frac{1}{t} + \frac{12(t - 1)}{n}\big)k + 1$.

We first note that it is easy to verify that this division is indeed envy-free; we complete the proof by arguing that no envy-free division can yield higher utilitarian welfare. Lemma~\ref{lem:u-ef-chosen} combined with Lemma~\ref{lem:u-condition} implies that there is no envy-free division in which contribution of the chosen players to the welfare exceeds $\frac{k}{t}$. It also follows from Lemma~\ref{lem:u-condition} that in any envy-free division the contribution of the Type B players to the welfare is bounded by $k(t-1)\cdot\frac{4}{n}$. Clearly, there is also no way to increase the contribution of the common players to the welfare to beyond $1$. We are thus left with the Type A players: observe that the only way to give a Type A player utility exceeding $\frac{4}{n}$ (without devouring too much of the high-values part) is to give her a piece intersecting both the high-values part and the common part. However, this is not profitable: suppose that we have such a division, in which some $\alpha$-fraction of the common part is given to players outside the set of common players. By our observations above, the utilitarian welfare of this division is bounded by $\big(\frac{1}{t} + \frac{12(t - 1)}{n}\big)k + (1-\alpha) + \alpha\cdot(1-\frac{4}{n}) < \big(\frac{1}{t} + \frac{12(t - 1)}{n}\big)k + 1$.
\end{proof}

The following lemma, combined with Lemma~\ref{u-ef-welfare}, completes the proof for Theorem~\ref{thm:duut}.

\begin{lemma}
By throwing away $k(t-1)$ intervals of the cake, we can achieve an envy-free division of the remaining cake with utilitarian welfare exceeding $k+2$.
\end{lemma}

\begin{proof}
Suppose that for each Type B player we throw away the one interval worth $\frac{1}{n}$ in the compensation part. We can now give each chosen player a piece containing all of her desired intervals; this collectively contributes $k$ to the utilitarian welfare. Also, since these players have received all their desired pieces, they will clearly envy no other players. Type A players now consider the pieces given to the chosen players as worth $\frac{4}{n}$; we can give each of them a piece of the same value from the compensation part. This contributes $\frac{4(2t-2)k}{n}$ to the welfare, and ensures that Type A players do not envy chosen players. Finally, the remaining players (Type B players and common players) share the common part such that each of them gets a piece of the same physical size. This guarantees the each common player value of $\frac{1}{\frac{n}{2}+(t-1)k}$ and each Type B player value of 
\begin{equation*}
\frac{1-\frac{4}{n}}{\frac{n}{2}+(t-1)k} = \frac{\frac{n-4}{n}}{\frac{n}{2} + \frac{n}{6} - \frac{k}{3}} = \frac{3n-12}{(2n-k)n} \geq \frac{3}{2n}
\end{equation*}
(since $(t-1)k = \frac{n}{6} - \frac{k}{3}$ and assuming $k\geq 8$). It thus follows that none of these players envy any other players (the other players clearly do not envy them), and their contribution to the welfare is $\frac{\frac{n}{2}+k(t-1)\cdot(1-\frac{4}{n})}{\frac{n}{2}+(t-1)k}$.
The utilitarian welfare of this division is therefore:
\begin{align*}
 k + \frac{8(t-1)k}{n} + \frac{\frac{n}{2}+k(t-1)(1-\frac{4}{n})}{\frac{n}{2}+k(t-1)} & = k + \frac{8(t-1)k}{n} + 1 - \frac{\frac{8(t-1)k}{n}}{n+2k(t-1)} \\
 &= k + 1 + \frac{8(t-1)k}{n}\cdot\left(1 - \frac{1}{n+2k(t-1)} \right) 
\end{align*}
However, since $6k(t-1) < n < 8k(t-1)- 1$ for $t > 2$, we have
\begin{align*}
k + 1 + \frac{8(t-1)k}{n}\cdot\left(1 - \frac{1}{n+2k(t-1)} \right) & > k + 1 + \frac{8(t-1)k}{n}\cdot\left(1 - \frac{1}{8k(t-1)} \right) \\
 & = k + 1 + \frac{8(t-1)k - 1}{n} > k + 2
\end{align*}
as stated.
\end{proof}

\section{Egalitarian Welfare}

\begin{theorem}\label{thm:dueg}
The egalitarian dumping paradox with $n$ players may get arbitrarily close to $\frac{n}{3}$, and this bound is asymptotically tight.
\end{theorem}

We will show that for every $k\in\mathbb{N}$, there exists a cake cutting instance with $n=3k+1$ players in which throwing away $k$ intervals of the cake can improve the egalitarian welfare of the best envy-free division by a factor arbitrarily close to $\frac{n}{3}$. The matching upper bound follows from Proposition~\ref{pro:de-pof}, combined with Theorem 5 of~\cite{AD10}, which shows an upper bound of $\frac{n}{2}$ on the Price of Envy-Freeness. 

To illustrate the main ideas of our lower bound construction, we begin with presenting the simple case of $n=4$. Fix some small $\epsilon>0$. We will have a cake with two parts: the ``main part'' and the ``last player'' part. In the main part, we have two ``blocks'' of four intervals: in both blocks, the first interval is of value $\frac{1}{4}$ to player 4 and the third interval is of value $\frac{1-\epsilon}{3}$ to player 3. The remaining intervals (second and fourth) of the first block are each of value $\frac{1+\epsilon}{4}$ to player 1, while those of the second block are each of value $\frac{1+\epsilon}{4}$ to player 2. The first block is followed by an interval of value $\epsilon$ to player 3; we denote this interval by $I$. The second block is followed by an interval of value $\frac{1-\epsilon}{3}$ to player 3. In the ``last player'' part we have two intervals of value $\frac{1}{4}$ to player 4; between these intervals there are two more intervals, one considered by player 1 as worth $\frac{1-\epsilon}{2}$, and the other considered by player 2 as worth $\frac{1-\epsilon}{2}$. Figure~\ref{fig:prefeg} illustrates these preferences graphically.

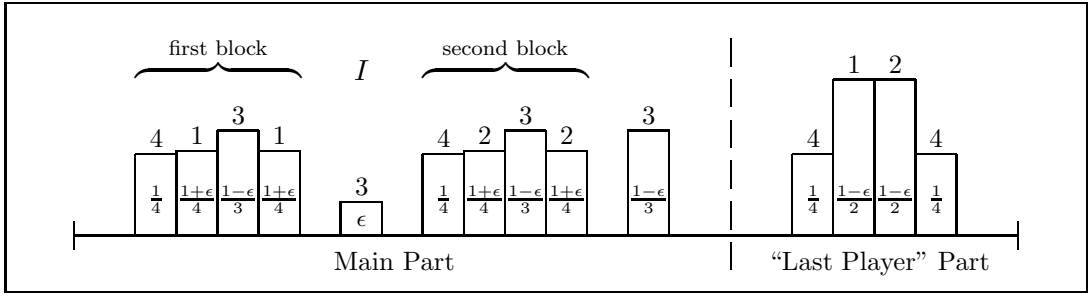
\begin{figure}[hbt]
\begin{center}
\framebox[1.1\width]{

\setlength{\unitlength}{9mm}
\begin{picture}(14.1,4)

\put(0,0.7){\line(1,0){13.8}}
\put(0,0.5){\line(0,1){0.4}}
\put(13.8,0.5){\line(0,1){0.4}}

\newsavebox{\rectFour}
\savebox{\rectFour}(0.6,1.2)[bl]{
	\put(0,0){\line(0,1){1.2}} \put(0,1.2){\line(1,0){0.6}}
	\put(0.6,0){\line(0,1){1.2}}
	\put(0.19,0.45){\footnotesize $\frac{1}{4}$}
}

\newsavebox{\rectFourEps}
\savebox{\rectFourEps}(0.6,1.25)[bl]{
	\put(0,0){\line(0,1){1.25}} \put(0,1.25){\line(1,0){0.6}}
	\put(0.6,0){\line(0,1){1.25}}
	\put(0.01,0.45){\footnotesize $\frac{1+\epsilon}{4}$}
}

\newsavebox{\rectThree}
\savebox{\rectThree}(0.6,1.55)[bl]{
	\put(0,0){\line(0,1){1.55}} \put(0,1.55){\line(1,0){0.6}}
	\put(0.6,0){\line(0,1){1.55}}
	\put(0.01,0.45){\footnotesize $\frac{1-\epsilon}{3}$}
}

\newsavebox{\rectTwo}
\savebox{\rectTwo}(0.6,2.3)[bl]{
	\put(0,0){\line(0,1){2.3}} \put(0,2.3){\line(1,0){0.6}}
	\put(0.6,0){\line(0,1){2.3}}
	\put(0.01,0.45){\footnotesize $\frac{1-\epsilon}{2}$}
}

\newsavebox{\rectEps}
\savebox{\rectEps}(0.6,0.5)[bl]{
	\put(0,0){\line(0,1){0.5}} \put(0,0.5){\line(1,0){0.6}}
	\put(0.6,0){\line(0,1){0.5}}
	\put(0.23,0.15){\footnotesize $\epsilon$}
}

\put(0.9,0.7){\usebox{\rectFour}} \put(1.5,0.7){\usebox{\rectFourEps}} \put(2.1,0.7){\usebox{\rectThree}} \put(2.7,0.7){\usebox{\rectFourEps}}
\put(1.11,2.0){\small 4} \put(1.71,2.05){\small 1} \put(2.31,2.35){\small 3} \put(2.91,2.05){\small 1}
\put(0.9,3.0){$\overbrace{\qquad\qquad\quad\ \ }^\text{first block}$}

\put(3.9,0.7){\usebox{\rectEps}} \put(4.11,1.3){\small 3}
\put(4.1,3.0){$I$}

\put(5.1,0.7){\usebox{\rectFour}} \put(5.7,0.7){\usebox{\rectFourEps}} \put(6.3,0.7){\usebox{\rectThree}} \put(6.9,0.7){\usebox{\rectFourEps}}
\put(5.31,2.0){\small 4} \put(5.91,2.05){\small 2} \put(6.51,2.35){\small 3} \put(7.11,2.05){\small 2}
\put(5.1,3.0){$\overbrace{\qquad\qquad\quad\ \ }^\text{second block}$}

\put(8.1,0.7){\usebox{\rectThree}}
\put(8.31,2.35){\small 3}

\multiput(9.6, 0.2)(0,0.6){6}{\line(0,1){0.4}}
\put(3.8,0.2){\small Main Part}
\put(10.15,0.2){\small ``Last Player'' Part}

\put(10.5,0.7){\usebox{\rectFour}} \put(11.1,0.7){\usebox{\rectTwo}} \put(11.7,0.7){\usebox{\rectTwo}} \put(12.3,0.7){\usebox{\rectFour}}
\put(10.71,2.0){\small 4} \put(11.31,3.1){\small 1} \put(11.91,3.1){\small 2} \put(12.51,2.0){\small 4}

\end{picture}} 
\caption{Preferences of all players when $n=4$. The number above each column denotes which player has that valuation.}\label{fig:prefeg}
\end{center}
\end{figure}


\begin{lemma}\label{lem:player4}
In every envy-free division of the above cake, player 4 has utility at most $\frac{1}{4}$.
\end{lemma}

\begin{proof}
Suppose otherwise, then it has to be that the piece of player 4 intersects at least two of her four desired intervals. If it intersects two of her first three desired pieces, we get that player 4 completely devours at least one of the blocks; however, each block is worth strictly more than $\frac{1}{2}$ to some player, and that player will envy player 4.

Thus, this interval must be contained in the ``last player'' section, and intersect player 4's third and fourth desired intervals. 
However, if this is the case, the piece of player $4$ is worth $\frac{1-\epsilon}{2}$ to both player $1$ and player $2$. To ensure envy-freeness, they both need to get a piece worth at least  $\frac{1-\epsilon}{2}$. Thus, if $\epsilon$ is small enough, player $1$ must get a piece containing the third interval of the first block, and player 2 must get a piece containing the third interval of the second block. Each of these pieces are therefore worth $\frac{1-\epsilon}{3}$ to player 3; this forces player 3 to get the rightmost of her desired intervals in order to avoid envy. Hence, the interval $I$ must be split between players 1 and 2. However, this way at least one of them will end up with a piece worth more than $\frac{1}{3}$ to player 3, making her envious; a contradiction.
\end{proof}

This implies that no envy-free division can have egalitarian welfare exceeding $\frac{1}{4}$. We now show that discarding one piece of the cake allows us to significantly increase the egalitarian welfare while maintaining envy-freeness.

\begin{lemma}
In the above cake, discarding the interval $I$ allows for an envy-free division with egalitarian welfare of $\frac{1 - \epsilon}{3}$.
\end{lemma}

\begin{proof} 
Suppose we discard the piece $I$. We can now allocate the entire ``last player'' part to player $4$, giving her utility $\frac{1}{2}$. In the main part, we give player 1 the entire first block, and player 2 the entire second block. Finally we give player 3 the interval following the second block. It is easy to verify that this division is indeed envy-free, and that its egalitarian welfare is $\frac{1-\epsilon}{3}$.
\end{proof}

We have shown a dumping paradox of $\frac{4(1-\epsilon)}{3}$ for the case of $n=4$ players. We will now generalize this construction, proving Theorem~\ref{thm:dueg}.

\begin{proof}[Proof of Theorem~\ref{thm:dueg}]
Similarly to the example above, we will have one player (player $n$) who can only get a big piece of cake (without causing envy) when some of the cake is discarded; this is the player creating the dumping paradox. Instead of the other three players, we will now have $3k$ players, divided into $k$ groups of $3$ players. The cake will again be composed of a main part and a ``last player'' part.

For every $1\leq j\leq k$, the players $3j-2,3j-1$ and $3j$ will form a ``group'', whose preferences resemble those of player 1,2 and 3 (respectively) in the case of $n=4$. For each such group, we will again have two blocks in the main part: in both blocks, the first interval is of value $\frac{1}{n}$ to player $n$ and the third interval is of value $\frac{1 - \epsilon}{3}$ to player $3j$. The remaining intervals (second and fourth) of the first of these blocks are each of value $\frac{1+\epsilon}{4}$ to player $3j-2$, while those of the second block are each of value $\frac{1+\epsilon}{4}$ to player $3j-1$. The two blocks are separated by an interval $I_j$ of value $\epsilon$ to player $3j$, and followed by an interval of value $\frac{1-\epsilon}{3}$ to player $3j$. 

In the ``last player'' part, we have $k+1$ intervals, each worth $\frac{1}{n}$ to player $n$. Separating the $j$-th and $j+1$-th of these intervals are two pieces: one of value $\frac{1 - \epsilon}{2}$ to player $3j-2$ and the other of value $\frac{1 - \epsilon}{2}$ to player $3j-1$. The reader is referred to Figure~\ref{fig:prefeg2} for a graphical representation of the players' preferences.
 
 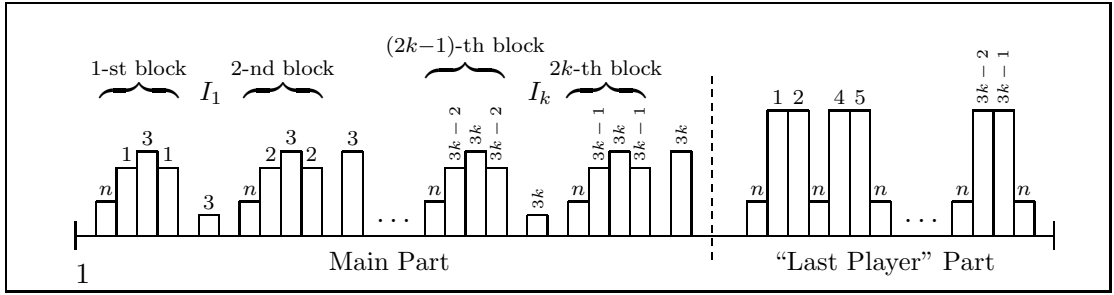
\begin{figure}[hbt]
\begin{center}
\framebox[1.1\width]{

\setlength{\unitlength}{9mm}
\begin{picture}(14.4,4)

\put(0,0.7){\line(1,0){14.3}}
\put(0,0.5){\line(0,1){0.4}}
\put(14.3,0.5){\line(0,1){0.4}}

\newsavebox{\bBlock}
\savebox{\bBlock}(1.2,1.24)[bl]{
	\put(0,0){\line(0,1){0.5}} \put(0.3,0){\line(0,1){1.0}}
	\put(0.6,0){\line(0,1){1.24}} \put(0.9,0){\line(0,1){1.24}}
	\put(1.2,0){\line(0,1){1.0}}
	
	\put(0,0.5){\line(1,0){0.3}} \put(0.3,1.0){\line(1,0){0.3}}
	\put(0.6,1.24){\line(1,0){0.3}} \put(0.9,1.0){\line(1,0){0.3}}
}

\newsavebox{\blockOne}
\savebox{\blockOne}(1.8,1.24)[bl]{
	\put(0,0){\usebox{\bBlock}}
	
	\put(1.5,0){\line(0,1){0.3}} \put(1.8,0){\line(0,1){0.3}}
	\put(1.5,0.3){\line(1,0){0.3}}
}

\newsavebox{\blockTwo}
\savebox{\blockTwo}(1.8,1.24)[bl]{
	\put(0,0){\usebox{\bBlock}}
	
	\put(1.5,0){\line(0,1){1.24}} \put(1.8,0){\line(0,1){1.24}}
	\put(1.5,1.24){\line(1,0){0.3}}
}

\newsavebox{\blockLast}
\savebox{\blockLast}(0.9,1.84)[bl]{
	\put(0,0){\line(0,1){0.5}} \put(0.3,0){\line(0,1){1.84}}
	\put(0.6,0){\line(0,1){1.84}}	\put(0.9,0){\line(0,1){1.84}}
	
	\put(0,0.5){\line(1,0){0.3}} \put(0.3,1.84){\line(1,0){0.6}}
}

\put(0.3,0.7){\usebox{\blockOne}} 
\put(0.35,1.3){\scriptsize $n$} \put(0.67,1.8){\scriptsize 1} \put(0.97,2.04){\scriptsize 3} \put(1.27,1.8){\scriptsize 1}
\put(0.22,2.7){$\overbrace{\qquad\ \ }^{1\text{-st block}}$}

\put(1.87,1.1){\scriptsize 3}
\put(1.81,2.7){\small $I_1$}

\put(2.4,0.7){\usebox{\blockTwo}}
\put(2.45,1.3){\scriptsize $n$} \put(2.77,1.8){\scriptsize 2} \put(3.07,2.05){\scriptsize 3} \put(3.37,1.8){\scriptsize 2}
\put(2.27,2.7){$\overbrace{\qquad\ \ }^{2\text{-nd block}}$}

\put(3.97,2.04){\scriptsize 3}

\put(4.4,0.9){\dots}

\put(5.1,0.7){\usebox{\blockOne}} 
\put(5.15,1.3){\scriptsize $n$}
\put(5.46,1.8){\begin{sideways}\tiny $3k-2$\end{sideways}}
\put(5.76,2.04){\begin{sideways}\tiny $3k$\end{sideways}}
\put(6.06,1.8){\begin{sideways}\tiny $3k-2$\end{sideways}}
\put(4.53,3.0){$\overbrace{\qquad\ \ }^{(2k-1)\text{-th block}}$}

\put(6.67,1.1){\begin{sideways}\tiny $3k$\end{sideways}}
\put(6.61,2.7){\small $I_k$}

\put(7.2,0.7){\usebox{\blockTwo}}
\put(7.25,1.3){\scriptsize $n$} 
\put(7.56,1.8){\begin{sideways}\tiny $3k-1$\end{sideways}}
\put(7.86,2.04){\begin{sideways}\tiny $3k$\end{sideways}}
\put(8.16,1.8){\begin{sideways}\tiny $3k-1$\end{sideways}}
\put(6.93,2.7){$\overbrace{\qquad\ \ }^{2k\text{-th block}}$}

\put(8.77,2.04){\begin{sideways}\tiny $3k$\end{sideways}}

\multiput(9.3,0.35)(0,0.2){14}{\line(0,1){0.1}}
\put(3.7,0.2){\small Main Part}
\put(10.2,0.2){\small ``Last Player'' Part}1

\put(9.6,0.7){\usebox{\blockLast}} \put(10.5,0.7){\usebox{\blockLast}}
\put(11.4,1.2){\line(1,0){0.3}} \put(11.7,0.7){\line(0,1){0.5}}
\put(9.65,1.3){\scriptsize $n$} \put(9.97,2.64){\scriptsize 1}
\put(10.27,2.64){\scriptsize 2} \put(10.55,1.3){\scriptsize $n$}
\put(10.87,2.64){\scriptsize 4} \put(11.17,2.64){\scriptsize 5} \put(11.45,1.3){\scriptsize $n$}

\put(11.9,0.9){\dots}

\put(12.6,0.7){\usebox{\blockLast}} \put(13.5,1.2){\line(1,0){0.3}} \put(13.8,0.7){\line(0,1){0.5}}
\put(12.65,1.3){\scriptsize $n$}
\put(12.96,2.64){\begin{sideways}\tiny $3k-2$\end{sideways}}
\put(13.26,2.64){\begin{sideways}\tiny $3k-1$\end{sideways}}
\put(13.55,1.3){\scriptsize $n$}

\end{picture}} \caption{Preferences of all players when $n=3k+1$}\label{fig:prefeg2}
\end{center}
\end{figure}

We can now argue, similarly to the case of $n=4$, that in any envy-free division of this cake, player $n$ gets a piece of value no more than $\frac{1}{n}$; this gives an upper bound of $\frac{1}{n}$ on the egalitarian welfare. Otherwise, one option is that player $n$ gets a piece containing a complete block from the main part; such a piece is worth more than $\frac{1}{2}$ to some player and makes her envy player $n$. The only other option is that player $n$ gets a piece intersecting two of her desired intervals from the ``last player'' part, and in this case similar reasoning as in Lemma~\ref{lem:player4} shows that again some player necessarily gets envious.

However, suppose that we discard of all the intervals $I_j$, $1\leq j\leq k$. Similarly to the case of $n=4$, this allows us to give the entire ``last player'' part to player $n$, for $1\leq j\leq k$ and $i\in\{1,2\}$ give the entire $(2j - i + 1)$-th block to player $3j-i$, and give each remaining interval to the single player $3j$ who desires it. This gives each player $3j-2$ or $3j-1$ a piece of value $\frac{1 +\epsilon}{2}$, each player $3j$ a piece of value $\frac{1-\epsilon}{3}$, and player $n$ a piece of value $\frac{k+1}{n} > \frac{1}{3}$. Again, it is easy to observe that this division causes no envy, and since its egalitarian welfare is $\frac{1-\epsilon}{3}$ this completes the proof, as we have shown an improvement of $\frac{(1-\epsilon)n}{3}$.
\end{proof}

\subsection{Tight Lower Bounds for $n\leq 4$}

For very small values of $n$, we can show that the upper bound of $\frac{n}{2}$ on the egalitarian dumping paradox is indeed tight.

\begin{theorem}
For $n\leq 4$ players, there are examples where the egalitarian dumping paradox is arbitrarily close to $\frac{n}{2}$.
\end{theorem}

\begin{proof}
For $n=2$, the upper bound implies that there is no egalitarian dumping paradox. It thus remains to prove the cases $n=3$ and $n=4$.

\paragraph{3 players.}
Fix some small $\epsilon > 0$. Player $1$ values the interval $(0,\epsilon)$ (her ``favorite interval'') as worth $\frac{1}{2} - \epsilon$, the interval $(1-\epsilon,1)$ (her ``second-favorite interval'') as worth $\frac{1}{2} - 2\epsilon$, and the interval $(\frac{2}{3},\frac{2}{3} + 3\epsilon)$ as worth $3\epsilon$. Players $2$ and $3$ value the entire cake uniformly.

We first note that in any complete envy-free division, none of the last two players (each of which has to get a piece of physical size at least $\frac{1}{3}$) can receive the rightmost part of the cake; such a piece is worth at least $\frac{1}{2} + \epsilon$ to player $1$ and will make her envy any other player who gets it. This implies that player $1$ must get the rightmost piece of the cake, and so the leftmost piece is given to some player $i\in\{2,3\}$. This leftmost piece (which again must be of physical size at least $\frac{1}{3}$) is worth $\frac{1}{2} - \epsilon$ to player $1$, and in order to avoid envy, her (rightmost) piece must be worth at least that much. We thus conclude that in any envy-free division player $1$ must get a piece containing the interval $(\frac{2}{3}+2\epsilon,1)$ (worth $\frac{1}{2}-\epsilon$ to her), leaving the two other players to share the remainder of the cake; each of them will get value of at most $\frac{1}{3} + \epsilon$, which is also the egalitarian welfare of such a division.

Now, consider the following partial division. We give the interval $(0,\epsilon)$ to player 1, the interval $(\epsilon,\frac{1}{2})$ to player $2$, $(\frac{1}{2},1-\epsilon)$ to player $3$, and discard the interval $(1-\epsilon,1)$. This is clearly an envy-free (partial) division, giving every player value of exactly $\frac{1}{2} - \epsilon$, which is therefore the egalitarian welfare; the ratio between these two welfare values is $\frac{3-6\epsilon}{2+6\epsilon}$, which approaches $\frac{n}{2} = \frac{3}{2}$ as $\epsilon \rightarrow 0$. 

\paragraph{4 players.}
Fix some $\epsilon > 0$. Player $1$ values the interval $(0,\epsilon)$ as worth $\frac{1}{2} - \epsilon$, the interval $(\frac{3}{4},\frac{3}{4} + 3\epsilon)$ as worth $3\epsilon$, and the interval $(1-\epsilon,1)$ as worth $\frac{1}{2} - 2\epsilon$. Player $2$ values $(\epsilon,2\epsilon)$ as worth $3\epsilon$, $(\frac{1}{4}-\epsilon,\frac{1}{4})$ as worth $\frac{1}{2} - 2\epsilon$, and $(\frac{1}{2},\frac{1}{2} + \epsilon)$ as worth $\frac{1}{2} - \epsilon$. Players $3$ and $4$ value the entire cake uniformly. We illustrate the preferences of players 1 and 2 in Figure~\ref{fig:pref}.

\begin{figure}[bth]
\begin{center}
\framebox[1.1\width]{

\setlength{\unitlength}{9mm}
\begin{picture}(12.3,4)

\put(0,1){\line(1,0){12}}
\put(0,0.8){\line(0,1){0.4}}
\put(12,0.8){\line(0,1){0.4}}
\put(-0.1,0.4){$0$} \put(11.9,0.4){$1$}

\put(3,0.9){\line(0,1){0.2}} 
\put(6,0.9){\line(0,1){0.2}} 
\put(9,0.9){\line(0,1){0.2}} 

\put(2.89,0.4){$\frac{1}{4}$} \put(5.89,0.4){$\frac{1}{2}$}  \put(8.89,0.4){$\frac{3}{4}$}

\put(0,1){\line(0,1){2.1}} \put(0,3.1){\line(1,0){0.6}} 
\put(0.6,1){\line(0,1){2.1}} \put(0.6,2){\line(1,0){0.6}} 
\put(1.2,1){\line(0,1){1}} 
\put(0.2,3.2){$1$} \put(0.8,2.1){$2$}
\put(0.07,1.05){\begin{sideways} \small $\frac{1}{2} - \epsilon$\end{sideways}}
\put(0.73,1.35){\small $3\epsilon$}

\put(2.4,1){\line(0,1){2.0}} \put(2.4,3){\line(1,0){0.6}} 
\put(3,1){\line(0,1){2.0}} 
\put(2.6,3.1){$2$}
\put(2.47,1.05){\begin{sideways} \small $\frac{1}{2} - 2\epsilon$\end{sideways}}

\put(6,1){\line(0,1){2.1}} \put(6,3.1){\line(1,0){0.6}} 
\put(6.6,1){\line(0,1){2.1}} 
\put(6.2,3.2){$2$}
\put(6.07,1.05){\begin{sideways} \small $\frac{1}{2} - \epsilon$\end{sideways}}

\put(9,1){\line(0,1){0.6}} \put(9,1.6){\line(1,0){1}} 
\put(10,1){\line(0,1){0.6}} 
\put(9.4,1.7){$1$}
\put(9.33,1.18){\small $3\epsilon$}

\put(11.4,1){\line(0,1){2.0}} \put(11.4,3){\line(1,0){0.6}} 
\put(12,1){\line(0,1){2.0}} 
\put(11.6,3.1){$1$}
\put(11.47,1.05){\begin{sideways} \small $\frac{1}{2} - 2\epsilon$\end{sideways}}

\end{picture}} \caption{Preferences of players $1$ and $2$. \label{fig:pref}}
\end{center}
\end{figure}
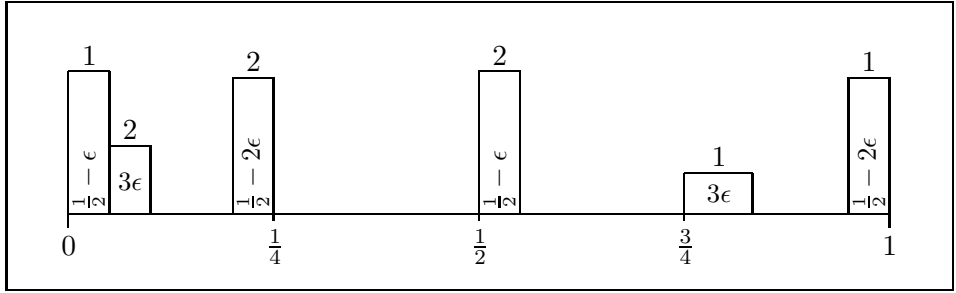

We first observe that if any player other than $1$ receives the rightmost piece, this player must receive a piece of physical size at least $\frac{1}{4}$; such a piece is worth $\frac{1}{2} + \epsilon$ for player $1$, and will make her envious of the player who got it. We conclude that the rightmost piece must therefore be given to player $1$. We further observe that the leftmost piece (which clearly cannot be also given to player $1$) must contain the interval $(0,\epsilon)$ and therefore worth $\frac{1}{2}-\epsilon$ to player 1. Thus, in order for player $1$ to avoid envy, she must get the rightmost piece, and this piece must contain the interval $(\frac{3}{4} + 2\epsilon,1)$.

We now consider the leftmost piece: If this piece is given to player $3$ or $4$, it must be of physical size at least $\frac{1}{4}$, and thus worth $\frac{1}{2} + \epsilon$ to player $2$, who will then envy that player. We thus conclude that player $2$ must receive the leftmost piece, and this piece must (strictly) contain the interval $(0,\frac{1}{4} - \epsilon)$. This implies that players $3$ and $4$ have only an interval contained in $(\frac{1}{4} - \epsilon,\frac{3}{4} + 2\epsilon)$ to share between them; in such a division, neither of them can get a piece worth (in her eyes) more than $\frac{1}{4} + 2\epsilon$, and so this is a bound on the maximum egalitarian welfare in any envy-free division of this cake.

Consider, in contrast, the following partial division: we give the interval $(0,2\epsilon)$ to player $1$, the interval $(2\epsilon,\frac{1}{2})$ to player $3$, the interval $(\frac{1}{2},\frac{1}{2} + \epsilon)$ to player $2$, $(\frac{1}{2}+\epsilon,1-\epsilon)$ to player $4$, and discard the interval $(1-\epsilon,1)$. It can be easily verified that this (partial) division is envy-free, and that it gives each of the players utility of at least $\frac{1}{2} - 2\epsilon$. The ratio between these two welfare values is $\frac{2 - 8\epsilon}{1 + 8\epsilon}$, which approaches $\frac{n}{2}=2$ as $\epsilon \rightarrow 0$. 
\end{proof}

\section{Pareto-Dominant Partial Divisions}

A division $x$ is said to \emph{Pareto dominate} another division $y$ if for all $i$, $u_i(x,i)\geq u_i(y,i)$, and at least one of these inequalities is strict; in other words, if at least one player does better in $x$ than in $y$, and no one does worse. $x$ \emph{strictly Pareto dominates} $y$ if for all $i$, $u_i(x,i) > u_i(y,i)$, i.e.~if \emph{everyone} is doing better in $x$.

We first show that starting from \emph{any} envy-free complete division it is impossible to strictly improve the utility of all players simultaneously. 

\begin{theorem}
Let $x$ be an envy-free complete division. Then there is no other division, partial or complete, that strictly Pareto dominates $x$.
\end{theorem}

\begin{proof}
Our proof hinges on the following observation, 
due to~\cite{AD10}:

\begin{quote}
\begin{itshape}
Let $y$ be a division such that $u_i(y,i) > u_i(x,i)$ for some $i\in[n]$. Since $i$ values any other piece in the division $x$ at most as much as her own, it has to be that in $y$, $i$ gets an interval that intersects pieces that were given to at least two different players in $x$ (possibly including $i$ herself).
\end{itshape}
\end{quote}
In other words, in order for a player $i$ to get a piece worth more than her piece in $x$, she must get at least one ``boundary'' (between two consecutive pieces) from $x$. Thus, a (partial or complete) division that strictly Pareto dominates $x$ must give (at least) one such boundary to each player. However, since $x$ is a connected division it contains only $n-1$ boundaries, one less than the number of players. 
\end{proof}



It it thus interesting that there {\em do} exist instances in which an envy-free partial division (non-strictly) Pareto dominates \emph{every} envy-free complete division. Moreover, in some cases the partial division improves the utility of almost all the players, and by a significant (constant) factor.

\begin{theorem}
For every $n > 2$, there exists a cake cutting instance with $n$ players and an envy-free partial division giving $n-2$ players \emph{twice} the value they would get in any envy-free complete division, while giving the remaining two players at least as much as they would get in any envy-free complete division.
\end{theorem}

\begin{proof}
Let $n > 2$, and fix some $0 < \epsilon < \frac{1}{n(n+1)}$. Consider the following valuations:
\begin{itemize}
\item Each player $1 \leq i \leq n-1$ (``focused players'') desires only the interval $(\frac{i}{n}-\epsilon,\frac{i}{n}+\epsilon)$, and considers it to be of value $1$.

\item Player $n$ assigns a uniform valuation to the entire cake.
\end{itemize}

\noindent Now, for envy-free complete division of the above cake, it must be that: 


\begin{enumerate}

\item \emph{Player $n$ gets a piece of physical size $\geq \frac{1}{n}$.}

Since we give away all the cake, some player must get a piece of physical size at least $\frac{1}{n}$; if player $n$ does not get such a piece, she will envy that player.

\item \emph{Player $n$ cannot get any piece containing some neighborhood of a point $\frac{i}{n}$ for $i\in[n-1]$.} \\
Because of the previous observation, if player $n$ gets such a piece then her piece contains the interval  $(\frac{i}{n}-\delta,\frac{i+1}{n}-\delta)$ for some $0 < \delta < 1$. Such a piece is always worth strictly more than $\frac{1}{2}$ to player $i$, and will make her envious.

Therefore, \emph{player $n$ must get a piece of the form $(\frac{i-1}{n},\frac{i}{n})$ for some $i\in[n]$}.

\item \emph{Every ``focused player" has to get a piece of physical size exactly $\frac{1}{n}$.} \\
First, it is clear that if some player $i\in[n-1]$ gets a piece of size $>\frac{1}{n}$, player $n$ will envy that player. On the other hand, we have that the players $[n-1]$ have to share cake of total physical size $\frac{n-1}{n}$; since none of them can get a piece of size larger than $\frac{1}{n}$, each of them must get a piece of size exactly $\frac{1}{n}$.


\end{enumerate}

From these observations we obtain that in any envy-free division, all the cuts are at points $\frac{i}{n}$ with $i\in[n-1]$; in such a division, player $n$ always has utility $\frac{1}{n}$, and every other player has utility $\frac{1}{2}$. 


Now consider the following partial division of the cake. First, give player $n$ the piece $(0,\frac{1}{n})$, and player $1$ the piece $(\frac{1}{n},\frac{2}{n} - 2\epsilon)$. Then give each player $2\leq i\leq n-1$ the next piece of size $\frac{1}{n} - \epsilon$, which is the interval $\left( i\cdot (\frac{1}{n} - \epsilon), (i+1)\cdot (\frac{1}{n} - \epsilon)\right)$. Finally, we throw away the (non-allocated) remainder.


In this division, player $n$ has value $\frac{1}{n}$, which is just as good as in any complete division. Similarly, player $1$ has value $\frac{1}{2}$, which again is as good as she can get in any complete division. Players $2$ through $n-1$, on the other hand, get each her entire desired interval; otherwise, the position of the right boundary of player $(n-1)$'s piece must be to the left of the point $\frac{n-1}{n} + \epsilon$. However, since we took $\epsilon < \frac{1}{n(n+1)}$, we have that the right boundary of player $(n-1)$'s piece is at 
\begin{equation*}
\left(\frac{1}{n} - \epsilon\right)\cdot n > \left(\frac{1}{n} - \frac{1}{n(n+1)}\right)\cdot n = \frac{n}{n+1} = \frac{n^2 - 1}{n(n+1)} + \frac{1}{n(n+1)} > \frac{n-1}{n} + \epsilon \;.
\end{equation*}
Therefore, each of these players gets a piece of value $1$, which is twice what they could get in any complete division.

Finally, it is clear that none of the players feel envy: Players $2$ through $n-1$ feel no envy (having gotten all they desire in the cake). Player $n$ feels no envy since she receives the physically-largest piece in the division. Player $1$ also feels no envy as her piece has value $\frac{1}{2}$, and so no other player could have gotten a piece with a larger value for her. 
\end{proof}

We note that the construction above can be also used to show a utilitarian dumping paradox arbitrarily close to $\frac{2(n-1)+\frac{2}{n}}{(n-1) + \frac{2}{n}}$ (one need only move the leftmost boundary in the partial division to $\frac{1}{n}-\epsilon$). While this is asymptotically inferior to the bound shown in Theorem~\ref{thm:duut}, this construction is much simpler, and works for as few as two players. In fact, for $n=2$ this construction coincides with the example given in the introduction, and moreover gives a provably tight lower bound: the dumping paradox of $\frac{3}{2}-\epsilon$ we obtain in this case matches the $n-\frac{1}{2}$ upper bound on the utilitarian Price of Envy-Freeness given in~\cite{CKKK09}.

\section{Discussion and Open Problems}

In this work, we have studied the dumping paradox and its possible magnitude. We have shown that the increase in welfare when discarding some of the cake can be substantial, moving from $1/n$ to $\Theta(1)$ for egalitarian welfare and from $\Theta(1)$ to $\Theta(\sqrt{n})$ for utilitarian welfare, and have shown a Pareto improvement that improves by a factor of two all but two players. In fact, in some cases discarding some of the cake can essentially eliminate the social cost associated with fair division. 
It is interesting to note that all of our lower bound constructions have an additional nice property --- no player desires any discarded piece more than her own piece. Thus, not only do players not envy each other, but they also do not feel much loss with any discarded piece.

Several problems remain open. First, we note that while our bounds for the utilitarian and egalitarian welfare functions are asymptotically tight, there are still constant gaps which await closure. With regards to Pareto improvement, we provided a construction where all but two players improve their utility by a factor of two. An interesting open problem is to see whether a stronger Pareto increase can be obtained.  Before we do so, however, we must first define the exact criteria by which we evaluate Pareto improvements.  Possible criteria include: the number of players that increase their utility, the largest utility increase by any player, and the total utility increase of the players (= utilitarian welfare). 

More important, perhaps, is that all of our results are existential in nature, but do not provide guidance on what to do in specific cases. It is thus of interest to develop algorithms to determine what, if any, parts of the cake it is best to discard in order to gain the most social welfare, for the different welfare functions.

Finally, our work joins other recent works~\cite{CLPP10,CLP11} that imply that leaving some cake unallocated may be a useful technique in fair division algorithms. Following this direction, it may be interesting to see if discarding of some cake may also help in finding socially-efficient envy-free connected divisions. Generalizing beyond fair division, it would be be interesting to see if such an approach, of intentionally forgoing or discarding some of the available goods, can also benefit other social interaction settings.

\bibliographystyle{alpha}
\bibliography	{cakes}

\begin{thebibliography}{CKKK09}

\bibitem[AD10]{AD10}
Yonatan Aumann and Yair Dombb.
\newblock The efficiency of fair division with connected pieces.
\newblock In {\em WINE}, pages 26--37, 2010.

\bibitem[BT95]{BT95}
Steven~J. Brams and Alan~D. Taylor.
\newblock An envy-free cake division protocol.
\newblock {\em The American Mathematical Monthly}, 102(1):9--18, 1995.

\bibitem[BT96]{BT96}
Steven~J. Brams and Alan~D. Taylor.
\newblock {\em Fair Division: From cake cutting to dispute resolution}.
\newblock Cambridge University Press, New York, NY, USA, 1996.

\bibitem[CKKK09]{CKKK09}
Ioannis Caragiannis, Christos Kaklamanis, Panagiotis Kanellopoulos, and Maria
  Kyropoulou.
\newblock The efficiency of fair division.
\newblock In {\em WINE}, pages 475--482, 2009.

\bibitem[CLP11]{CLP11}
Ioannis Caragiannis, John Lai, and Ariel Procaccia.
\newblock Towards more expressive cake cutting.
\newblock In {\em IJCAI: International Joint Conferences on Artificial
  Intelligence}, 2011.

\bibitem[CLPP10]{CLPP10}
Yiling Chen, John Lai, David~C. Parkes, and Ariel~D. Procaccia.
\newblock Truth, justice, and cake cutting.
\newblock In {\em AAAI}, 2010.

\bibitem[DS61]{DS61}
L.~E. Dubins and E.~H. Spanier.
\newblock How to cut a cake fairly.
\newblock {\em The American Mathematical Monthly}, 68(1):1--17, Jan 1961.

\bibitem[EP84]{EP84}
S.~Even and A.~Paz.
\newblock A note on cake cutting.
\newblock {\em Discrete Applied Mathematics}, 7(3):285 -- 296, 1984.

\bibitem[EP06]{EP06}
Jeff Edmonds and Kirk Pruhs.
\newblock Cake cutting really is not a piece of cake.
\newblock In {\em SODA '06: Proceedings of the seventeenth annual ACM-SIAM
  symposium on Discrete algorithm}, pages 271--278, New York, NY, USA, 2006.
  ACM.

\bibitem[MIBK03]{MIBK03}
Malik Magdon-Ismail, Costas Busch, and Mukkai~S. Krishnamoorthy.
\newblock Cake-cutting is not a piece of cake.
\newblock In {\em STACS}, pages 596--607, 2003.

\bibitem[Mou04]{Mou04}
Herv\'e~J. Moulin.
\newblock {\em Fair Division and Collective Welfare}.
\newblock Number 0262633116 in MIT Press Books. The MIT Press, 2004.

\bibitem[Pro09]{Pro09}
Ariel~D. Procaccia.
\newblock Thou shalt covet thy neighbor's cake.
\newblock In {\em IJCAI}, pages 239--244, 2009.

\bibitem[RW98]{RW98}
Jack Robertson and William Webb.
\newblock {\em Cake-cutting algorithms: Be fair if you can}.
\newblock A K Peters, Ltd., Natick, MA, USA, 1998.

\bibitem[Ste49]{Ste49}
H.~Steinhaus.
\newblock Sur la division pragmatique.
\newblock {\em Econometrica}, 17(Supplement: Report of the Washington
  Meeting):315--319, Jul 1949.

\bibitem[Str80]{Str80}
Walter Stromquist.
\newblock How to cut a cake fairly.
\newblock {\em The American Mathematical Monthly}, 87(8):640--644, 1980.

\bibitem[SW03]{SW03}
Jiri Sgall and Gerhard~J. Woeginger.
\newblock A lower bound for cake cutting.
\newblock In {\em ESA}, pages 459--469, 2003.

\end{thebibliography}

\end{document}